\documentclass[acmsmall, screen, nonacm]{acmart}

\usepackage{comment}
\usepackage{algorithm}

\usepackage{algpseudocode}
\usepackage[capitalize]{cleveref}
\usepackage{braket}
\usepackage{thm-restate}

\newcommand{\cG}{\mathcal{G}}

\begin{document}


\title{(Quantum) complexity of testing signed graph clusterability}

\author{Kuo-Chin Chen}
\affiliation{ 
 \institution{Hon Hai Quantum Computing Research Center}
 \city{Taipei}
 \country{Taiwan}}
\email{Jim.KC.Chen@foxconn.com}

\author{Simon Apers}
\affiliation{ 
 \institution{Universit\'e de Paris, CNRS, IRIF}
 \city{Paris}
 \country{France}}
\email{smgapers@gmail.com}

\author{Min-Hsiu Hsieh}
\affiliation{ 
 \institution{Hon Hai Quantum Computing Research Center}
 \city{Taipei}
 \country{Taiwan}}
\email{min-hsiu.hsieh@foxconn.com}

\begin{abstract}
This study examines clusterability testing for a signed graph in the bounded-degree model. Our contributions are two-fold.
First, we provide a quantum algorithm with query complexity $\tilde{O}(N^{1/3})$ for testing clusterability, which yields a polynomial speedup over the best classical clusterability tester known \cite{adriaens2021testing}. 
Second, we prove an $\tilde{\Omega}(\sqrt{N})$ classical query lower bound for testing clusterability, which nearly matches the upper bound from \cite{adriaens2021testing}.
This settles the classical query complexity of clusterability testing, and it shows that our quantum algorithm has an advantage over any classical algorithm.

\end{abstract}

\maketitle


\section{Introduction} 
  
Property testing \cite{goldreich2010property, ron2010algorithmic} deals with the setting where we wish to distinguish between objects, e.g.~functions \cite{ailon2007estimating, alon2003testing_s, kaufman2008algebraic} or graphs \cite{alon2003testing, alon2002testing, alon2008testing, alon2005every, apers2018quantum, alon2002testing_g}, that satisfy a predetermined property and those that are far from satisfying this property.  
For certain properties, this relaxed setting allows for algorithms to query only a small part of (sometimes huge) data sets.
Indeed, the goal in property testing is to design so-called \emph{property testers} to solve a property testing problem within sublinear time complexity.
Property testing has been studied in many settings, such as computational learning theory \cite{goldreich1998property, bansal2004correlation, ron2008property, gur2021sublinear, diakonikolas2007testing, fischer2004testing}, quantum information theory \cite{montanaro2013survey, buhrman2008quantum, datta2013smooth, bubeck2020entanglement, nagy2018quantum, ambainis2016efficient}, coding theory \cite{goldreich2005short, trevisan2004some, jutla2004testing, kaufman2006testing, lin2022c, panteleev2022asymptotically, dinur2022locally}, and so on. 
This witnesses the significant attention that property testing has drawn from the academic community. 

An interesting setting is that of graph property testing. In the \emph{dense graph model}, it was shown that a constant number of queries are needed to test a wide range graph partition properties \cite{goldreich1998property}, including $k$-colorability, $\rho$-clique, and $\rho$-cut for any fixed $k\geq 2$ and $\rho > 0$. 
For comparison, in the \emph{bounded-degree model} \cite{goldreich1997property} similar graph properties such as bipartiteness and expansion testing require sublinear $\tilde{\Theta}(\sqrt{N})$ classical queries.
Moreover, some graph properties even have a (trivial) $\Omega(N)$ query complexity, as Ref. \cite{bogdanov2002lower} showed for 3-colorability in the bounded-degree model.
While there have been numerous studies on testing graph properties, there has been little work on testing the properties of \emph{signed} graphs. 

A signed graph is a graph where each edge is assigned a positive or a negative label. 
They can be applied to model a variety of problems including correlation clustering problems \cite{bansal2004correlation, demaine2006correlation}, modeling the ground state energy of Ising models \cite{kasteleyn1963dimer}, and social network problems \cite{harary1953notion,leskovec2010signed,tang2016survey}. 
Signed graphs have different properties than unsigned graphs.  
One of these is the important property of clusterability, which was introduced by Davis \cite{davis1967clustering} to describe the correlation clustering problem.  
We call a signed graph clusterable if it can be decomposed into several components such that (1) the edges in each component are all positive, and (2) the edges connecting the vertices belonging to different components are all negative.  
This property is equivalent to not having a ``bad cycle'', which is a cycle with exactly one negative edge \cite{davis1967clustering}.  
An algorithm for testing clusterability in the bounded-degree model with only $\tilde O(\sqrt{N})$ was proposed in \cite{adriaens2021testing}.  
The optimality of this clusterability tester was left as an open question.  
Here, we prove that any classical algorithm requires at least~$\tilde{\Omega}(\sqrt{N})$ queries to test clusterability, showing that the tester from \cite{adriaens2021testing} is nearly-optimal.

As a natural extension of past studies, we are interested in whether quantum computing can provide any advantages in testing clusterability for signed graphs.  
To the best of our knowledge, we are not aware of any previous work on the quantum advantage for testing the properties of signed graphs.
However, in work by Ambainis, Childs and Liu \cite{ambainis2011quantum}, a quantum speedup for testing bipartiteness and expansion of bounded-degree graphs was shown.
We adopt these techniques to obtain a quantum algorithm for testing clusterability in signed graphs. 
More specifically, we combine their quantum approach with the classical property testing techniques provided by Adriaens and Apers \cite{adriaens2021testing} to obtain a quantum algorithm for testing the clusterability of bounded-degree graphs in time $\tilde{O}(N^{1/3})$.
This outperforms the $\tilde{O}(\sqrt{N})$ query complexity of the classical tester in \cite{adriaens2021testing} (which is optimal by our lower bound).
We leave optimality of the quantum algorithm for testing clusterability as an open question.
Indeed, settling the quantum query complexity of property testing in the bounded-degree graph model has been a long open question, and even for the well-studied problem of bipartiteness testing no matching lower bound is known \cite{ambainis2011quantum}.

\subsection{Overview of Main Results} 

Here we formally state our main results (precise definitions are deferred to \cref{sec:prelims}).
First, we prove a lower bound on the classical query complexity of clusterability testing for a signed graph. 
{ 
\renewcommand{\thetheorem}{\ref{Theorem:lower_bound}} 
\begin{theorem}[Restated] 
Any classical clusterability tester with error parameter $\epsilon=0.01$ must make at least $\sqrt{N}/10$ queries.
\end{theorem} 
\addtocounter{theorem}{-1} 
} 
Up to polylogarithmic factors this matches the upper bound from \cite{adriaens2021testing}, thus proving that their clusterability tester is optimal in the classical computing regime.
However, taking inspiration from this classical clusterability tester, we reduce the clusterability testing problem to a collision finding problem which can be solved faster by quantum computing.  
As a result, we propose a quantum clusterability tester with a query complexity $\tilde{O}(N^{1/3})$.  
  
{ 
\renewcommand{\thetheorem}{\ref{theorem_quantum_tester}} 
\begin{theorem}[Restated] 
We propose a quantum clusterability tester with query complexity~$\tilde{O}(N^{1/3})$.
\end{theorem} 
\addtocounter{theorem}{-1} 
} 
\noindent
This improves over the classical lower bound, implying a quantum advantage over classical algorithms for testing clusterability.
  

\subsection{Technical contributions} 
A sketch of the proof of our two results is given in this section. 
The first result is the classical query lower bound for testing clusterability.
While the bound follows the blueprint of the lower bound for bipartiteness testing by Goldreich and Ron \cite{goldreich1997property}, we have to deal with a number of additional complications in the signed graph setting.

The main idea of the lower bound is to show that, with less than $\sqrt{N}/10$ queries, we cannot distinguish two families of graphs: one family $\mathcal{G}_1^N$ that is $\epsilon$-far from clusterable, and another family $\mathcal{G}_2^N$ that is clusterable. 
The design of these two families of graphs must take into account two constraints.  
The first constraint is that the graphs in $\mathcal{G}_{2}^{N}$ cannot contain a bad cycle, while those in $\mathcal{G}_{1}^{N}$ must have at least one bad cycle, even if we remove $\epsilon N d$ edges of the graph.  
This ensures that $\mathcal{G}_2^N$ is clusterable, while $\mathcal{G}_1^N$ is far from clusterable.
The second constraint relates to the fact that both graph should be locally indistinguishable.
This requires for instance that vertices in each graph in both families are incident to the same number of positive and negative edges.
If in addition we can ensure that each cycle in these graphs contains many edges with a constant probability, then we can use this to show that no algorithm can distinguish the graphs in these two families with $o(\sqrt{N})$ queries.
Indeed, we show that these two families of graphs are indistinguishable with less than $\sqrt{N}/10$ queries as follows.
First, we propose two random processes ${P}_1$ and ${P}_2$, one generates a uniformly random graph in $\mathcal{G}_{1}^{N}$, and the other generates a uniformly random graph in $\mathcal{G}_{2}^{N}$.  
Specifically, ${P}_{\alpha}$ for $\alpha\in \{1,2\}$ takes a query given from an algorithm as input and returns a vertex while ``on-the-fly'' or ``lazily'' constructing a graph from $\mathcal{G}^{N}_{\alpha}$.  
In other words, $P_{\alpha}$ simulates how an algorithm interacts with a graph sampled uniformly in $\mathcal{G}^{N}_{\alpha}$.  
We observe that these random processes are statistically identical if the answer to each query is not found in the past answers or queries, which is equivalent to not finding a cycle when exploring a graph. 
Second, we demonstrate that the probability of these random processes being statistically identical is greater than $1/4$ within $\sqrt{N}/10$ queries. In other words, no classical algorithm can distinguish between these two families with a probability exceeding $3/4$ within $\sqrt{N}/10$ queries to the input graph.

Our second result is a quantum algorithm for clusterability testing with a better query complexity.
To this end, we reduce the main procedure in the algorithm proposed by Adriaens and Apers \cite{adriaens2021testing} to a collision finding algorithm.
This collision finding problem can then be solved using the quantum collision finding algorithm, similar to \cite{ambainis2011quantum}.  
The main idea is that if we implement several random walks on the positive edges of a graph that is far from clusterable, then there exists a negative edge between the vertices belonging to distinct random walks with a constant probability.  
We define finding such a negative edge between random walks as finding a collision, a process that can be solved by using a quantum collision finding algorithm.
This yields a quantum speedup for clusterability testing.

\section{Preliminaries} \label{sec:prelims}
This section contains two parts.
\cref{Terminologies} defines some of the basic terminology associated with the graphs used in this paper. 
In \cref{Clusterability}, we introduce the graph clusterability testing problem.
  
\subsection{Terminology} \label{Terminologies} 
  
A graph $G = (V, E)$ is a pair of sets.  
The elements in $V=[N]$ are vertices, and the elements in $E$, denoted by edges, are paired vertices. 
The vertices $v \in V$ and $u \in V$ of an edge $(v, u)\in E$ are the endpoints of $(v, u)$, and $(v, u)$ is incident to $u$ and $v$. 
The vertices $u$ and $v$ are adjacent if there exist an edge $(v, u)\in E$.
The number of edges incident with $v$, denoted by $d(v)$, is the degree of a vertex, and the maximum degree among the vertices in $G$ is the degree of the graph $G(V, E)$. 

Given a graph $G = (V, E)$, a walk is a sequence of edges $( (v_1,v_2) , (v_2,v_3), \cdots,  (v_{J-1},v_J) )$ where $(v_j,v_{j+1}) \in E$ for all $1 \leq j \leq J-1$ and $v_j \in V$ for all $1 \leq j \leq J$. 
This walk can also be denoted as a sequence of vertices $(v_1,v_2, \ldots, v_J)$. 
A trail is a walk in which all edges are distinct.
A cycle is a non-empty trail in which only the first and last vertices are equal.
A Hamiltonian cycle is a cycle of a graph in which every vertex is visited exactly once.

A signed graph $G=( 
V, E, \Sigma)$ consists of the vertex set $V$, the edge set $E \subseteq V \times V$, and a mapping $\Sigma: E \rightarrow\{+,-\}$ that indicates the sign of each edge.  
We say that a signed graph $G=(V, E, \Sigma)$ is clusterable if we can partition vertices into components such that (i) every edge that connects two vertices in the same components is positive, and (ii) every edge that connects two vertices in different components is negative.

\subsection{Clusterability testing for signed graphs} \label{Clusterability} 
  
We can easily modify the usual graph query model to signed graphs.  
Given a signed graph $G$ with maximum degree $d$, the bounded-degree graph model is defined as follows. 
A query is a tuple $(v,i)$ where $v\in[N]$ is a vertex in the graph and $i \in[d]$. 
The oracle answers this query with (i) the $i^{\text {th}}$ neighbor of the vertex $v$ if the degree of $v$ is larger than $i$ (otherwise it returns an error symbol), and (ii) the sign of the corresponding edge.
  
Property testing in the bounded-degree model is described as follows.  
Given oracle access to a graph $G$ with degree bound $d$ and $|V|=N$, we wish to distinguish whether the graph $G$ satisfies a certain property, or whether it is $\epsilon$-far from any graph having that property, where $\epsilon \in (0,1]$ is an error parameter. 
Here we say that two graphs $G$ and $G'$ are $\epsilon$-far from each other if we have to add or remove at least $\epsilon N d$ edges to turn $G$ into $G'$. 
The specific case of clusterability testing is defined formally as follows.
\begin{definition} 
A clusterability testing algorithm is a randomized algorithm that has query access to a signed graph $G(V,E,\Sigma)$ with $|V|=N$ and maximum degrees $d$. Given an error parameter $\epsilon$, the algorithm behaves as follows:  
\begin{itemize} 
   \item If $G$ is clusterable, then the algorithm should accept with probability at least $2/3$. 
   \item If $G$ is $\epsilon$-far from clusterable, then the algorithm rejects with probability at least $2/3$. 
\end{itemize} 
\end{definition}

\section{Main Results and Proofs} 

In this section, we give the formal statements and proofs of our two main results -- a classical query lower bound for clusterability testing and a quantum clusterability tester. 
In \cref{CTheorem:lower_bound}, we first give the classical query lower bound of $\Omega(\sqrt{N})$ for clusterability testing.  
This result claims the optimality of the classical clusterability tester in \cite{adriaens2021testing}. 
In \cref{QCT}, we provide a quantum clusterability tester with query complexity $\tilde{O}(N^{1/3})$ which outperforms the classical clusterability tester in \cite{adriaens2021testing}.

\subsection{Classical query lower bound for testing clusterability}\label{CTheorem:lower_bound} 
In this section, we derive a classical query lower bound for the clusterability testing problem. Specifically, we show that testing the clusterability of a signed graph with $N$ vertices 
requires at least $\sqrt{N}/10$ queries.  
  
\begin{theorem}\label{Theorem:lower_bound} 
Given a signed graph $G$ with $N$ vertices, testing clusterability of $G$ with error parameter $\epsilon = 0.01$ requires at least $ \sqrt{N}/10$ queries. 
\end{theorem}

\begin{proof} 

The proof consists of three main steps.
First, we construct two families of graphs denoted as $\mathcal{G}_{1}^{N}$ and $\mathcal{G}_{2}^{N}$, each possessing specific desirable properties. In particular, we require that most graphs within $\mathcal{G}_{1}^{N}$ are at least $0.01$-far from being clusterable, while graphs within $\mathcal{G}_{2}^{N}$ are inherently clusterable.
The construction and analysis of these families is deferred to \cref{sec:4_1}.

To prove \cref{Theorem:lower_bound}, we illustrate the interaction between an arbitrary $T$-query clusterability testing algorithm $\mathcal{A}$ and a graph $g$ uniformly sampled from $\cG^N_{\alpha}$ as follows:

For all $t \leq T$, each query $q_{t}$ is represented as a tuple $\left(v_{t}, i_{t}\right)$, and the answer to $q_{t}$ is denoted as $a_t$, where $v_{t}, a_{t} \in [N]$ and $i_{t} \in[6]$. It is crucial to note that each query $q_{t}$ corresponds to an edge in $g$, specifically the edge $(v_{t},a_{t})$.
We additionally denote a list of tuples $h=[(q_1,a_1),(q_2,a_2),\ldots,(q_t,a_t)]$ as the query-answer history. 
This history is generated by the interaction between $\mathcal{A}$ and $g$ in the following manner:
For each $t \leq T$, $\mathcal{A}$ maps $h$ to $q_{t+1}$ and ultimately to either accept or reject for $t=T$.
For a given history $h=[(q_1=(v_1,i_1),a_1),\ldots,(q_t=(v_t,i_t),a_t)]$, we say that a vertex $u$ is in $h$ if $u=v_{t'}$ or $u=a_{t'}$ for some $t'\in [t]$.

Secondly, we introduce two processes, denoted as $P_{\alpha}$ for $\alpha\in \{1,2\}$, which simulate how an algorithm $\mathcal{A}$ interacts with a graph sampled uniformly from $\mathcal{G}^{N}_{\alpha}$.
To be more specific, we consider that $\mathcal{A}$ interacts with a graph $g$ sampled from $\mathcal{G}^N_\alpha$ and generates the query-answer history $h$.
We must have that the graph $g$ is uniformly distributed in ${\mathcal{G}}^{N}_{\alpha,h} \subseteq \mathcal{G}^{N}_{\alpha}$, where ${\mathcal{G}}^{N}_{\alpha,h}$ includes all graphs that produce the query-answer history $h$ during interactions with $\mathcal{A}$.
Therefore, if $\mathcal{A}$ makes a query $q_{t+1} \notin \{q_i\}_{i=1}^{t}$ to a graph uniformly sampled from ${\mathcal{G}}^{N}_{\alpha,h}$, we can determine that the answer corresponds to a certain vertex $u\in[N]$ with a specific probability denoted as $\mathrm{p}_u$.
The random processes $P_{\alpha}$ are precisely defined to return the answer $u$ with the corresponding probability $\mathrm{p}_u$ when responding to the query $q_{t+1}$ (initiated by $\mathcal{A}$) and considering the history $h$.
As a result, these two random processes, $P_{\alpha}$, interact with $\mathcal{A}$, providing responses to $\mathcal{A}$'s queries while simultaneously constructing a graph uniformly distributed in $\mathcal{G}_{\alpha}^N$.
The description and analysis of these random processes are deferred to \cref{sec:4_2}.



In the third part, we demonstrate that no algorithm can with high probability differentiate between query-answer histories generated during the interactions of $\mathcal{A}$ with $P_1$ and $P_2$ while making less than $\sqrt{N}/10$ queries. 
To prove such indistinguishability, we examine the distribution of query-answer histories of length $T$ denoted as $\mathbf{D}
_{\alpha}^{\mathcal{A}}$, where each element in $\mathbf{D}_{\alpha}^{\mathcal{A}}$ is generated through the interactions of $\mathcal{A}$ and $P_{\alpha}$. The statistical difference between $\mathbf{D}_1^{\mathcal{A}}$ and $\mathbf{D}_2^{\mathcal{A}}$ is defined as follows:
  
\begin{align*} 
\frac{1}{2} \cdot \sum_{x}\left|\operatorname{Prob}\left[\mathbf{D}_1^{\mathcal{A}}=x\right]-\operatorname{Prob}\left[\mathbf{D}_2^{\mathcal{A}}=x\right]\right|,
\end{align*}
where $x$ is some query-answer history of length $T$.  
We then provide an upper bound on this statistical difference in the following lemma. The proof of this lemma is a modification of the proof of Lemma 7.4 in \cite{goldreich1997property}, and we defer its proof to \cref{sec:4_3}.

\begin{lemma}[{based on \cite{goldreich1997property}, Lemma 7.4)}] \label{lemma:statistical_diff} 
Let $\delta<\frac{1}{2}, T \leq \delta \sqrt{N}$ and $N \geq 40T$. For every algorithm $\mathcal{A}$ that uses $T$ queries, the statistical distance between $\mathbf{D}_{1}^{\mathcal{A}}$ and $\mathbf{D}_{2}^{\mathcal{A}}$ is at most $10 \delta^{2}$. 
\end{lemma}

Finally, we establish \cref{Theorem:lower_bound} through a proof by contradiction. Let us assume the existence of a clusterability tester $\mathcal{A}$ that requires only $\sqrt{N}/10$ queries. Consequently, we can infer that the probability of $\mathcal{A}$ accepting a graph from $\mathcal{G}^{N}_2$ is at least $2/3$.
By referring to \cref{lemma:statistical_diff}, we determine that the statistical difference between $\mathbf{D}_1^{\mathcal{A}}$ and $\mathbf{D}_2^{\mathcal{A}}$ is at most $10\delta^2=1/10$ where $\delta$ is set $1/10$ for a $\sqrt{N}/10$-query algorithm. Hence, $\mathcal{A}$ accepts a graph distributed uniformly in $\mathcal{G}^{N}_1$ with a probability of at least $2/3 - 1/10 > 0.4$.

Furthermore, as indicated by \cref{proposition_G_1}, more than $99\%$ of the graphs in $\mathcal{G}^{N}_1$ are at least $0.01$-far from being clusterable. Consequently, by the definition of a clusterability tester, we can conclude that $\mathcal{A}$ accepts a graph distributed uniformly in $\mathcal{G}^{N}_1$ at most $0.99\cdot 1/3 + 0.01 < 0.35$. This contradicts the earlier finding that $\mathcal{A}$ accepts a graph distributed uniformly in $\mathcal{G}^{N}_1$ with a probability of at least $0.4$.
Hence, we can deduce that there does not exist a clusterability tester capable of distinguishing between a graph sampled from $\mathcal{G}^{N}_1$ and $\mathcal{G}^{N}_2$ using only $\sqrt{N}/10$ queries, and the theorem follows.
\end{proof}

\subsection{Quantum clusterability tester} \label{QCT}

\begin{algorithm} 
\caption{Quantum clusterability tester}\label{Algo:quantum_tester} 
\begin{algorithmic}[1] 
\Require  
Oracle access to a signed graph $G(V,E,\Sigma)$ with $N$ vertices and degree bound $d$; an accuracy parameter $\epsilon \in (0,1]$.  
\For{$O(1/\epsilon)$ times} 
\State Pick a vertex $s\in V$ randomly. 
\State Let $K=O\left(\sqrt{N} \operatorname{poly}(\log N / \epsilon ) \right)$, $L=\operatorname{poly}\left(\log N/\epsilon\right)$, $n=K L$, and $k=\Theta(L)$. 
\State Adopt \cref{proposition:random_variable} to construct $k$-wise independent random variables $b_{ij}$ taking values 
\State in~$[2d]$ for $i\in [K]$ and $j\in [L]$. 
  
\State Run the quantum collision finding algorithm in \cref{lemma:collision_finding} with the following setting:  
\begin{itemize}
    \item $X:= [K] \times [L]$; $Y$ is the set of tuples $(v,v_{\text{neb}})$ where $v\in V$ and $v_{\text{neb}}$ is the set of vertices adjacent to $v$. 
    \item A function $f$ that take $(i, j) \in X$ as input, and return the endpoint of a random walk that starts at $s$ with random coin flips $\left(b_{i 1}, \ldots, b_{i j}\right)$. 
    \item  Symmetric binary relation $R\subseteq Y\times Y$ defined as follows:
    $$
    \left((v, v_{\text{neb}}),(v^{\prime}, v_{\text{neb}}^{\prime})\right) \in R \text{ iff } (v\in v_{\text{neb}}^{\prime} \text{ and the edge between } v \text{ and } v' \text{ is negative).} 
    $$  
\end{itemize}
\If{quantum collision finding algorithm finds a collision} 
\State \Return false 
\EndIf 
\EndFor 
  
\State \Return true 
\end{algorithmic} 
\end{algorithm}

In this section, we present our second result: a quantum clusterability tester (\cref{Algo:quantum_tester}) with a query complexity of $O\left(N^{1/3} \operatorname{poly}\left(\log N/\epsilon\right)\right)$.
We begin by introducing the quantum clusterability tester, followed by the proof of its correctness in \cref{theorem_quantum_tester}.

\cref{Algo:quantum_tester} takes a signed graph $G(V,E,\Sigma)$ with $N$ vertices and a bound on the maximum degree $d$, along with an accuracy parameter $\epsilon \in (0,1]$, as input. The goal is to determine whether $G(V,E,\Sigma)$ is clusterable or $\epsilon$-far from clusterable. 
The algorithm consists of four major steps.

First, \cref{Algo:quantum_tester} randomly selects a vertex $s\in V$.
Second, it constructs random variables that determine the direction of movement in each step of these random walks.
To achieve this, we need to prepare $O(K\cdot L)$ random variables ($K$ and $L$ are defined in \cref{Algo:quantum_tester}); however, we can derandomize and reduce the number of random bits from $O(K\cdot L)$ to $O(L)$ because  \cref{Algo:quantum_tester} only depends on each pair of walks that are selected from $K$ random walks.
Therefore, it is sufficient to construct $k$-wise independent random variables $b_{ij}$ mapping to $[2d]$ for $i\in[K]$ and $j\in[L]$, where $k=\Theta(L)$.
This construction can be realized by the following proposition.

\begin{proposition}\label{proposition:random_variable} 
(\cite{alon1986fast}, Proposition 6.5) 
Let $n+1$ be a power of 2 and $k$ be an odd integer such that $k \leq n$. In this scenario, there exists a uniform probability space denoted as $\Omega=\{0,1\}^{m}$, where $m=1+\frac{1}{2}(k-1) \log_{2}(n+1)$. Within this space, there exist $k$-wise independent random variables, represented as $\xi_{1}, \ldots, \xi_{n}$ over $\Omega$, such that $\operatorname{Pr}\left[\xi_{j}=1\right]=\operatorname{Pr}\left[\xi_{j}=0\right]=\frac{1}{2}$.

Moreover, an algorithm exists that, when provided with $i \in \Omega$ and $1 \leq j \leq n$, can compute $\xi_{j}(i)$ in a computational time of $O(k \log n)$.
\end{proposition}

Third, we define a function $f$ that implements random walks according to these random variables. 
$f$ returns the endpoint of a random walk and the neighborhood of this endpoint. 
Specifically, we let $X= \{1, \ldots, K\} \times\{1, \ldots, L\}$, and $Y$ be the set of tuples $(v,v_{\text{neb}})$ where $v\in \{1, \ldots, N\}$ and $v_{\text{neb}}$ is the set of vertices adjacent to $v$.  
Then, we define the function $f$ as follows. 
$f$ takes $(i, j) \in X$ as input, then it runs a random walk according to the random variables $\left(b_{i 1}, \ldots, b_{i j}\right)$ such that (i) this walk starts at $s$ and (ii) each edge in this walk is positive.  
The function $f$ finally returns a tuple $(v,v_{\text{neb}})$.

Fourth, we define the symmetric binary relation $R\subseteq Y\times Y$ such that $\left((v, v_{\text{neb}}),(v', v_{\text{neb}}')\right) \in R$ iff (i) $v\in v_{\text{neb}}'$, and (ii) the edge between $v$ and $v'$ is negative.  
In other words, detecting a collision is equivalent to detecting a bad cycle. 
The last step is to detect two distinct elements $x_1, x_2 \in X$ such that $(f(x_1),f(x_2))$ satisfies the symmetric binary relation $R$.  
The collision finding problem can be improved by a quantum collision finding algorithm proposed by Ambainis \cite{ambainis2011quantum} as follow.

\begin{lemma} 
\label{lemma:collision_finding} 
(\cite{ambainis2011quantum}, Theorem 9) 
Given a function $f:X\rightarrow Y$, and a symmetric binary relation $R \subseteq Y \times Y$ which can be computed in $\operatorname{poly} (\log |Y|)$ time steps where $X$ and $Y$ are some finite sets, we denote a collision by a distinct pair $x, x^{\prime} \in X$ such that $\left(f(x), f\left(x^{\prime}\right)\right) \in R$.
There exists a quantum algorithm that can find a collision with a constant probability when a collision exists, and always returns false when there does not exist a collision. 
The running time of the quantum algorithm is $O\left(|X|^{2 / 3} \cdot \operatorname{poly} (\log |Y|)\right)$. 
\end{lemma} 

By this lemma we can identify a bad cycle within $K$ random walks, with a query complexity of $O(|X|^{2/3})=O((K\cdot L)^{2/3})=O\left((\sqrt{N}\operatorname{poly}(\log N/\epsilon))^{2/3}\right)=O(N^{1/3}\operatorname{poly}(\log N/\epsilon)).$
Next, we establish the correctness of this algorithm and present its time complexity in the following theorem.
\begin{theorem} 
\label{theorem_quantum_tester}  
\cref{Algo:quantum_tester} is a quantum algorithm that tests the clusterability of a signed graph with query complexity and running time $O(N^{1/3}\operatorname{poly}(\log N/\epsilon))$.
\end{theorem} 
Following our first result, we conclude that our quantum clusterability tester outperforms any classical clusterability tester.  

\begin{proof}

First we prove that \cref{Algo:quantum_tester} is indeed a clusterability tester.
When $G$ is clusterable, signifying the absence of one bad cycle, \cref{Algo:quantum_tester} fails to discover a collision. Consequently, it returns true.
On the contrary, when $G$ is $\epsilon$-far from clusterable, the assertion in Claim 14 from \cite{adriaens2021testing} suggests that the algorithm can pinpoint a bad cycle within the sampled random walks with a constant probability. 
This leads \cref{Algo:quantum_tester} to return false with a constant probability.

To bound the time complexity (and hence query complexity), we need to bound the following quantities:
\begin{itemize}
    \item The time required to evaluate the $k$-wise independent random variables.
    \item The time required to evaluate $f$.
    \item The number of queries required to find a collision.  
\end{itemize}

For the first requirement, it takes $O\left(\operatorname{poly} (\log N/\epsilon)\right)$ time to evaluate a $k$-wise independent random variable, as indicated by \cref{proposition:random_variable}.
Moving to the second requirement, it is evident that each evaluation of $f$ consumes time $\operatorname{poly}(\log N/\epsilon)$ since $f$ is a procedure implementing a random walk, and the length of the walk is $L\in\operatorname{poly}(\log N/\epsilon)$.
Concerning the last requirement, we are aware that detecting a collision requires $O\left(N^{1 / 3} \operatorname{poly}\left(\log N/\epsilon\right)\right)$ time, as derived from \cref{lemma:collision_finding}.
In conclusion, the query and time complexity of \cref{Algo:quantum_tester} is $O\left(N^{1 / 3} \operatorname{poly}\left(\log N/\epsilon\right)\right)$.  
\end{proof}

\section{Proof details}
In this section, we detail the construction and lemmas in \cref{Theorem:lower_bound}.
In \cref{sec:4_1}, we generate two distinct families of graphs, each exhibiting different property of clusterability.
In \cref{sec:4_2}, we introduce two random processes that interact with an arbitrary algorithm $\mathcal{A}$ during the generation of graphs selected uniformly from the aforementioned families.
In \cref{sec:4_3}, we demonstrate that the statistical difference of query answer histories produced by $\mathcal{A}$ and these two random processes is bounded by the number of queries.

\subsection{Graph construction}\label{sec:4_1}

Here we detail the construction and analysis of the graph families $\mathcal{G}_1^N$ and $\mathcal{G}_2^N$.
 
\subsubsection{Construction of two families of signed graphs $\cG^N_{\alpha}$.}

We detail the construction of two families of signed graphs denoted as $\mathcal{G}^N_{1}$ and $\mathcal{G}^N_{2}$. In both families, each signed graph consists of $N$ vertices, where $N$ is a multiple of 10. Each vertex $v$ is assigned a label $p_v$ chosen from the set $\{0, 1, \ldots, 9\}$ in such a way that there are exactly $N/10$ vertices for each possible label.

For the edge set, we embed them in a manner such that each vertex is incident to precisely $6$ edges.
We construct edge sets based on cycles associated to a permutation $\sigma=(r_1\ r_2\ \ldots \ r_L)$, where $0 \leq L \leq 9$ and $r_l \in \{0, 1, \ldots, 9\}$ are distinct for $1\leq l\leq L$,
With some abuse of notation, we also denote by $\sigma$ the bijective function $\sigma:\{r_1, r_2, \ldots, r_L\} \rightarrow \{r_1, r_2, \ldots, r_L\}$ defined as
\begin{align*}
\sigma(r_l)=
\begin{cases}
r_{l+1}   & \text { if } l<L. \\
r_{1}   & \text { if } l=L. \\
\end{cases}
\end{align*}
With this notation, we define a family $\mathcal{D}^{\sigma}$ such that each member of this family is a union of cycles satisfying two properties: (i) the union of cycles contains all vertices in $[N]$ labeled with values from $r_0$ to $r_L$, and (ii) for each cycle $(v_1, v_2, \ldots, v_J)$ in the union of cycles, the label for each vertex must satisfy $p_{v_{j+1}} = \sigma(p_{v_j})$ for $1 \leq j \leq J$ (where we set $v_{J+1}=v_1$).
We then employ these cycles to define the edge sets for the graphs in the family $\mathcal{G}^N_{\alpha}$. 
See \cref{Fiproposition_G_1} for an illustration for $\mathcal{G}^{40}_{1}$.

\begin{description}
  \item [For $\cG^N_{1}$:]

Each graph in $\mathcal{G}^N_1$ consists of one Hamiltonian cycle and two unions of cycles (we later comment on the particular choice of $\sigma$'s):
\begin{itemize}
    \item 
    The Hamiltonian cycle $\in \mathcal{D}^{\sigma^{\text{1st}}}$ with $\sigma^{\text{1st}} = (0\ 1\ 2\ 3\ 4\ 5\ 6\ 7\ 8\ 9)$.
    We call this the arc cycle.
    All of its edges are positive, and we refer to these edges as arc edges.
    \item
    One union of cycles $\in \mathcal{D}^{\sigma^{\text{2nd}}}$ with $\sigma^{\text{2nd}} = (2\ 4\ 6\ 0\ 8\ 1\ 3\ 7\ 9\ 5)$,\footnote{
The choice of $\sigma^{\text{2nd}}$ and $\sigma^{\text{3rd}}$ is not unique; we only require that these edge sets are disjoint when we fix the label of each vertex. This forbids picking for example $\sigma^{\text{2nd}}=(0\ 2\ 4\ 6\ 8\ 1\ 3\ 5\ 7\ 9)$, since the edges connecting vertices labeled $9$ and $0$ can be found in both $\mathcal{D}^{\sigma^{\text{1st}}}$ and $\mathcal{D}^{\sigma^{\text{2nd}}}$, meaning they are not disjoint. However, we could replace $\sigma^{\text{2nd}}$ with $(2\ 6\ 4\ 0\ 8\ 1\ 3\ 7\ 9\ 5)$, where exchanging $6$ and $4$ would not violate the disjoint property.} with each of its edges being positively signed. We call these edges connecting edges.
    \item
    A second union of cycles $\in \mathcal{D}^{\sigma^{\text{3rd}}}$ with $\sigma^{\text{3rd}} = (1\ 6\ 3\ 8\ 5\ 0\ 7\ 2\ 9\ 4)$, and all edges are negatively signed.
\end{itemize}

  \item [For $\cG_{2}^{N}$:]
In the family of graphs $\mathcal{G}_2^N$, each graph contains one Hamiltonian cycle and $12$ unions of cycles:
\begin{itemize}
    \item 
    The Hamiltonian cycle $\in \mathcal{D}^{\sigma^{\text{1st}}}$ with each of its edges negatively signed. 
    \item 
    There are ten additional unions of cycles $\in \mathcal{D}^{\sigma^{s}}$ with $\sigma^{s} = (s)$ for $s$ taking values in the set $\{0, 1, \ldots, 9\}$. These edges are positive.
    \item 
    The last two unions of cycles are disjoint. One belongs to $\mathcal{D}^{\sigma^{10}}$ with $\sigma^{10} = (0\ 2\ 4\ 6\ 8)$. The other belongs to $\mathcal{D}^{\sigma^{11}}$ with $\sigma^{11} = (1\ 3\ 5\ 7\ 9)$.
    These edges are positive.
\end{itemize}
\end{description}

\begin{figure}[ht]
\centering
\includegraphics[width=13cm]{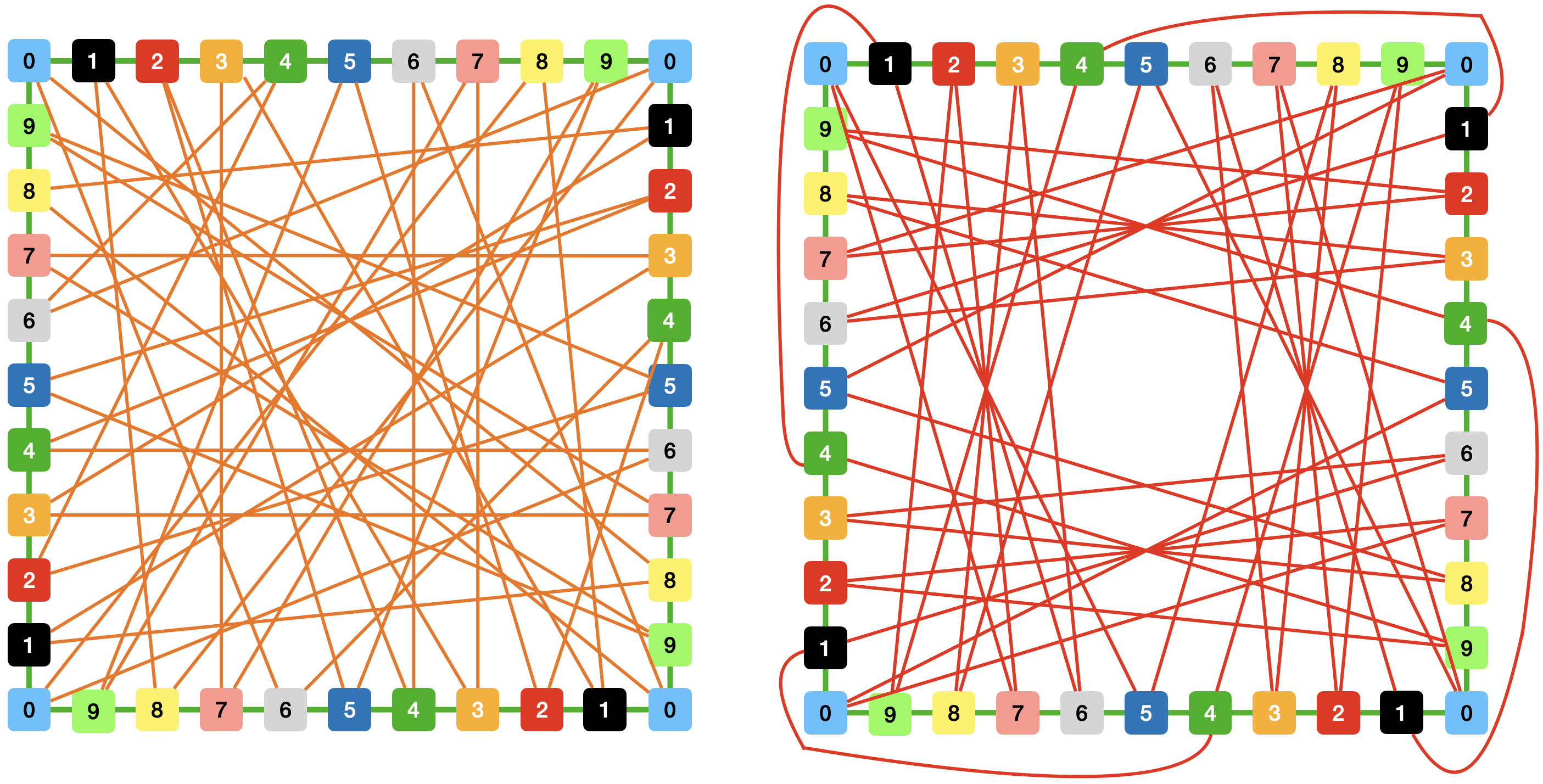}
\caption{
\label{Fiproposition_G_1}
An instance of $\cG_{1}^{40}$. The green lines indicate the edges in one Hamiltonian cycle belonging to $\mathcal{D}^{\sigma^{\text{1st}}}$, the orange lines indicate the edges in one union of cycles belonging to $\mathcal{D}^{\sigma^{\text{2nd}}}$, and the red lines indicate the edges in one union of cycles belonging to $\mathcal{D}^{\sigma^{\text{3rd}}}$.
 }
\end{figure}

In every graph within $\mathcal{G}^N_{\alpha}$, each vertex is incident to precisely six edges, and these edges are labeled according to the following convention:
For a pair of adjacent vertices, represented as $v_j$ and $v_{j+1}$ for $1\leq j\leq J-1$, within any cycle $(v_1, v_2, \ldots, v_J)$, we label the edge connecting them as $k$ for $v_m$ and as $k+1$ for $v_{m+1}$, for some $k\in \mathbb{N}$. This labeling effectively associates an orientation to the cycle.
More specifically, in a graph from $\mathcal{G}^N_1$:
\begin{itemize}
\item 
The edges in the Hamiltonian cycle from $\mathcal{D}^{\sigma^{\text{1st}}}$ are labeled with $1$ and $2$.
\item
For the edges in the union of cycles $\in \mathcal{D}^{\sigma^{\text{2nd}}}$, we use labels $3$ and $4$.    
\item
For the edges in the union of cycles $\in \mathcal{D}^{\sigma^{\text{3rd}}}$, we use labels $5$ and $6$.
\end{itemize}

In the case of a graph from $\mathcal{G}^N_2$:
\begin{itemize}
\item
The edges in the Hamiltonian cycle corresponding to $\mathcal{D}^{\sigma^{\text{1st}}}$ are labeled as $5$ and $6$.
\item
For the edges in the union of cycles within $\mathcal{D}^{\sigma^{s}}$ for $s$ ranging from $0$ to $9$, we assign labels $1$ and $2$.
\item
For the edges in the union of cycles $\in\mathcal{D}^{\sigma^{10}}$ and $\mathcal{D}^{\sigma^{11}}$, we label them with $3$ and $4$.
\end{itemize}

\subsubsection{Clusterability of $\cG_{\alpha}^N$}
We initially observe that all graphs within $\cG_2^N$ are clusterable  with the following rationale. 
All vertices with even (odd, respectively) labeling are interconnected via positive edges. Consequently, two connected components emerge: one component comprises all vertices with even labeling, and the other includes all vertices with odd labeling, each with positive edges. We further note that these two components can only be connected through negative edges. 
Hence, any graph within $\mathcal{G}_{2}^{N}$ satisfies the clusterability definition. As a result, all graphs in the second family are clusterable.

Regarding the graphs in $\mathcal{G}_{1}^{N}$, we will demonstrate that they are at least $0.01$-far from being clusterable with probability at least $1- \exp (-\Omega (N))$ in the following \cref{proposition_G_1}.

\begin{proposition} \label{proposition_G_1} 
The graphs in $\mathcal{G}_{1}^{N}$ are $0.01$-far from clusterable with probability at least $1- \exp (-\Omega (N))$.
\end{proposition}

\begin{proof}

We commence our proof by providing a description of the random process used to uniformly generate a graph denoted as $g$ from $\mathcal{G}_{1}^{N}$. We begin by constructing the set of vertices~$[N]$ and refer to the resulting (empty) graph as $sg$. 
The graph $sg$ is equipped with its edges set through a three-step process:

\begin{enumerate}
    \item \textbf{(One Hamiltonian cycle):}
    In the first step, we uniformly select a Hamiltonian cycle from all possible Hamiltonian cycles on the vertices set $[N]$ and assign each vertex a label from the set $\{0,1, \ldots, 9\}$, based on the rule of cycles $\in\mathcal{D}^{\sigma^{\text{1st}}}$. All edges constructed in this step are positive.
    \item 
    \textbf{(Second edge set from }$\mathcal{D}^{\sigma^{\text{2nd}}}$\textbf{):} 
    In the second step, we repeat the following processes $N$ times:
    \begin{enumerate}
        \item 
        Select an arbitrary vertex $u_i$ that lacks an edge labeled as $3$ where the index $i\in[N]$ represents the label of iteration.
        \item 
        Uniformly select a vertex $v_i$ from a set that includes all vertices labeled as $\sigma^{\text{2nd}}(p_{u_i})$ and that lack an edge labeled as $4$.
        \item 
        Add the edge $(u_i,v_i)$.
    \end{enumerate}
    This adds an edge set from $\mathcal{D}^{\sigma^{\text{2nd}}}$.
    We make these edges positive.
    \item 
    \textbf{(Third edge set from }$\mathcal{D}^{\sigma^{\text{3rd}}}$\textbf{):}
    Similar to the previous procedure, we add an edge set from $\mathcal{D}^{\sigma^{\text{3rd}}}$.
    We make these edges negative.
\end{enumerate}

We call the resulting graph $g$, and note that $g$ is a uniformly random element from $\mathcal{G}^{N}_1$.
We proceed to observe that each graph $g$ in $\mathcal{G}_{1}^{N}$ is inherently non-clusterable. Indeed, unless we remove arc edges from the Hamiltonian cycle, every negative edge of a cycle in $\mathcal{D}^{\sigma^{\text{3rd}}}$ contributes to a bad cycle.
We show that, with high probability over the random graph in $\mathcal{G}^N_1$, removing less than $0.01 d N = 0.06N$ edges cannot make the graph clusterable.

More precisely, we will establish that after removing less than $0.06N$ arc edges, with high probability, all vertices remain connected through (positive) connecting edges in $\mathcal{D}^{\sigma^{\text{2nd}}}$, and so the graph cannot be clustered.
To prove this, let us delve into a more detailed description of the random process used to generate a graph $g$.

In the first step, we construct a Hamiltonian cycle and eliminate $x < 0.06N$ arc edges, resulting in a graph with $x$ components. 
There are $C^{N}_{x} = \binom{N}{x}$ possible possibilities for these $x$ components.

During the first iteration of the second step, we select the arbitrary vertex $u_0$ from the component with the fewest vertices and designate this component as $C$. 
It becomes evident that, in the first iteration of step 2(c), the edge $(u_0, v_0)$ connects component $C$ to another component, with a probability exceeding $1/2$. 
Consequently, the number of components in the graph $sg$ decreases by $1$, and the number of vertices in $C$ increases with a probability greater than $1/2$.

In the subsequent iterations, we select the vertex $u_i$ for $2 \leq i \leq N$ based on the following rule: 
If the number of vertices labeled as $\sigma^{\text{2nd}}(p_{v_{i-1}})$ and lacking edges labeled as $4$ within the component $C$ is fewer than the number of vertices labeled as $\sigma^{\text{2nd}}(p_{v_{i-1}})$ not in $C$, then we set the vertex $u_{i}$ equal to $v_{i-1}$. Subsequently, in 2(b) and 2(c), the process embeds an edge connecting $u_{i}$ to some vertex $v_{i}$ that is not a resident in $C$ with a probability greater than $1/2$.
Otherwise, we choose $u_{i}$ from any arbitrary vertex labeled as $\sigma^{\text{2nd}}(p_{v_{i-1}})$ and not in $C$. Subsequently, in 2(b), the process selects $v_{i}$ in $C$ with a probability greater than $1/2$, as $C$ has more vertices capable of connecting with $u_{i}$ than the set of vertices not in $C$.

Consequently, the probability of the graph having more than one component can be bounded by the probability of obtaining fewer than $x$ heads when flipping $N$ unbiased coins. This probability can be bounded as follows:

\begin{align*}
\sum^{x}_{i=2}
C^{N}_{i}
\left(\frac{1}{2}\right)^{N-i} 
<
2^{N \cdot H(0.06)}
2^{-N}
2^{x}
<
2^{N(-1+0.06+H(0.06))}
,
\end{align*}
where $H(p)=-p\log(p)-(1-p)\log(1-p)$ is the (binary) entropy function. 

At this point we removed only $x<0.06N$ edges and we are permitted to remove an additional $0.06N - x$ connecting edges from $sg$.
This corresponds to $0.06N - x$ tests where the coin flips tails (thus reducing the number of components by $1$) can be taken into account for flips resulting in heads.
In other words, the condition of having fewer than $x$ heads can be extended to having fewer than $x+0.06N - x$ heads when flipping $N$ unbiased coins.
Consequently, the probability that the resulting graph has more than one component, when $0.06N - x$ connecting edges are removed, can be bounded as:

\begin{align*}
\sum^{x+(0.06N-x)}_{i=2}
C^{N}_{i}
\left(\frac{1}{2}\right)^{N-i} 
<
2^{N(-1+0.06+H(0.06))}.
\end{align*}
Given that there are $C^{N}_{x} < 2^{NH(0.06)}$ possible ways to construct $x$ components in the first step, we can confidently assert that, after implementing step (2), all vertices in $sg$ are interconnected by positive edges with a probability of at least $1 - \exp^{-\Omega(N)}$, even in cases where $0.06N$ positive edges (comprising of $x$ arc edges and $0.06N-x$ connecting edges) were removed.
The negative edges are present in the edge set in $\mathcal{D}^{\sigma^{\text{3rd}}}$, and each of them generates a bad cycle under the condition that only one component (only positive edges inside) is left after completing the second step in the process. In other words, under this condition, we must remove all negative edges to make this graph clusterable.
Consequently, the lemma follows.
\end{proof} 

\subsection{Random processes}
\label{sec:4_2}

Here we construct and analyze the random processes that play a key role in our lower bound.
The first part describes the interaction of a random process $P_{\alpha}$ with an algorithm $\mathcal{A}$.
The second part proves that $P_{\alpha}$ indeed generates a graph uniformly within $\mathcal{G}^{N}_{\alpha}$, as further elucidated in  \cref{proposition:random_process}.

We will begin by defining the random process $P_1$, which involves two stages. The first stage will explain how $P_1$ interacts with an arbitrary $T$-query algorithm $\mathcal{A}$. The second stage will elaborate on how $P_1$ constructs a graph uniformly sampled from  $\mathcal{G}_{1}^{N}$.

\begin{description} 
\item[First stage of \boldmath$P_1$:] 
Given a query-answer history $h=[(q_1, a_1), (q_2, a_2), \ldots, (q_{t-1}, a_{t-1})]$ for $t\leq T$, we define a set of vertices $X_{p,i}$, which contains all vertices labeled with $p$ in the history and lacking edge $i$. We also use the notation $n_p$ to represent the count of vertices in this history that are labeled $p$.
For each query $q_{t}=(v_{t}, i_{t})$ made by $\mathcal{A}$, the actions of $P_1$ are defined as follows:
\begin{enumerate} 
 \item  
 If $v_{t}$ is not in $h$, then $P_1$ labels $v_{t}$ with a number $p\in \{0,1,\cdots,9\}$ with a probability of $\frac{(N/10)-n_{p}}{N-\left(\sum_{p=0}^{9} n_{p}\right)}$. Subsequently, $P_1$ answers this query as described in (2) below.
 
 \item If $v_{t}$ belongs to $h$, there are two possible scenarios:
   
\begin{enumerate} 
  
\item 
If we can find the edge corresponding to $q_{t}$ in $h$, then $P_1$ responds with the vertex connected to this edge. In other words, there exists an edge $(v_t,u)$ in $h$ such that $(v_t,u)$ is labeled as $i_t$ for vertex $v_t$, and $P_1$ responds with $u$. The query-answer history remains unchanged in this case.

\item 
If the edge corresponding to $q_{t}=(v_t,i_t)$ does not exist in $h$, we follow these steps:
Suppose, without loss of generality, that $i_t=1$. We set the label $\sigma^{\text{1st}}(p_{v_t})$ as $\overline{p}$ and $\overline{i}=i_t+1=2$. 
$P_1$ decides whether to uniformly select a vertex from $X_{\overline{p},\overline{i}}$ by flipping a coin with bias $\frac{|X_{\overline{p}, \overline{i}}|}{N/10\ -\ n_{\overline{p}}\ +\ |X_{\overline{p}, \overline{i}}|}$ or to uniformly select a vertex not present in $h$, and assigns the label $\overline{p}$ to it. In either case, $P_1$ responds with the selected vertex $u$, and the edge $(v_t,u)$ is signed positively. Subsequently, this edge $(v_t,u)$ is added to the query-answer history $h$.

For the other case ($i_t=2,3,4,5,6$), $P_1$ acts in a similar manner as described above, except for the assignment for $\overline{p}$, the assignment for $\overline{i}$, and the sign of the added edge. The added edge is positively signed for $i_t=2,3,4$, and negatively signed for $i_t=5,6$. For $\overline{i}$, it is set to $i_t+1$ for $i_t=3,5$ and to $i_t-1$ for $i_t=2,4,6$. The assignment for $\overline{p}$ is as follows:

$$
\begin{cases}
\overline{p}\leftarrow(\sigma^{\text{1st}})^{-1}(p_{v_t}) 
& \text{for } i_t=2   \\ 
\overline{p}\leftarrow \sigma^{\text{2nd}}(p_{v_t})
& \text{for } i_t=3   \\ 
\overline{p}\leftarrow(\sigma^{\text{2nd}})^{-1}(p_{v_t})
& \text{for } i_t=4   \\ \overline{p}\leftarrow\sigma^{\text{3rd}}(p_{v_t}) & \text{for } i_t=5   \\ \overline{p}\leftarrow(\sigma^{\text{3rd}})^{-1}(p_{v_t}) & \text{for } i_t=6   \\ 
\end{cases}
$$

\end{enumerate} 

\end{enumerate} 
  
\item [First stage of \boldmath$P_2$:] 
$P_2$ follows a similar process to $P_1$, with the only differences being the assignment for $\overline{p}$ as follows:
$$
\begin{cases}
\overline{p}\leftarrow p_{v_t} & \text{for } i_t=1   \\ 
\overline{p}\leftarrow p_{v_t} & \text{for } i_t=2   \\ 
\overline{p}\leftarrow p_{v_t}+2 \ (\text{mod } 10)
& \text{for } i_t=3   \\ 
\overline{p}\leftarrow p_{v_t}-2 \ (\text{mod } 10)& \text{for } i_t=4   \\ 
\overline{p}\leftarrow \sigma^{\text{1st}}(p_{v_t})& \text{for } i_t=5  \\
\overline{p}\leftarrow(\sigma^{\text{1st}})^{-1}(p_{v_t})& \text{for } i_t=6   \\ 
\end{cases}
$$

\item [Second stage of \boldmath$P_1$:] 

After answering all of these queries and generating a query-answer history $\left[\left(q_{1}, a_{1}\right), \ldots,\left(q_{T}, a_{T}\right)\right]$, $P_1$ proceeds with the following processes:

\begin{enumerate} 
 \item
 Uniformly selecting a feasible way to embed the edges in $h$ on a cycle. The embedding of these edges adhere to the following conditions:
 Each vertex is assigned a cycle position, i.e., an integer in $\{0, \ldots, N-1\}$, in a manner that ensures each vertex labeled with $p\in\{0,\ldots,9\}$ is positioned at a position $x$ such that $p\equiv x$ (mod $10$).
 This assignment implies that all acr edges (labeled $1,2$) are placed on the cycle, and edges labeled $3,4,5,6$ are excluded from the cycle.
 \item Randomly positioning all other vertices on the cycle, ensuring that each vertex $v$ with label $p_{v}$ is positioned at a position $x$ such that $p_{v}\equiv x$ (mod $10$). Subsequently, all cycle edges are assigned a positive sign.
 \item In the end, uniformly selecting a feasible way to embed the edges sets in 
 $\mathcal{D}^{\sigma^{\text{2nd}}}$ and $\mathcal{D}^{\sigma^{\text{3rd}}}$. All edges in the edges sets 
 $\in\mathcal{D}^{\sigma^{\text{2nd}}}$ assigned a positive sign, while all edges in the edges sets 
 $\in\mathcal{D}^{\sigma^{\text{3rd}}}$  are assigned a negative sign.
 \end{enumerate} 

\item [Second stage of \boldmath$P_2$:] 
$P_2$ follows a process similar to that of $P_1$, with few distinctions:
In (2), we assign each cycle edge a negative sign. 
In (3), $P_2$ uniformly selects a feasible way to embed the edges sets $\in\mathcal{D}^{\sigma^{s}}$ for $s\in\{0,1,\cdots,11\}$, and assigns positive signs to these edges.

\end{description}  
We will show that the above two processes uniformly generate a graph in the corresponding family in the next lemma.  
  
\begin{proposition} \label{proposition:random_process} 
For every algorithm $\mathcal{A}$ that uses $T$ queries and for each $\alpha \in\{1,2\}$, the process $P_{\alpha}$ uniformly generates graphs in $\mathcal{G}_{\alpha}^{N}$ when interacting with $\mathcal{A}$.  
\end{proposition}

\begin{proof} 

We will use induction to prove this lemma. Consider that every probabilistic algorithm can be viewed as a distribution of deterministic algorithms. Therefore, it is sufficient to prove this lemma for any deterministic algorithm $\mathcal{A}$.
The base case (i.e., $T=0$) is correct because the query-answer history is empty, and the second stage in the process $P_{\alpha}$ uniformly generates a graph in $\mathcal{G}_{\alpha}^{N}$.
We assume that the claim is true for $T-1$, and we will prove that the claim is also true for $T$. Let $\mathcal{A}'$ be the algorithm defined by stopping $\mathcal{A}$ before it asks the $T^{th}$ query.
By the inductive assumption, we know that $P_{\alpha}$ uniformly generates graphs in $\mathcal{G}_{\alpha}^{N}$ when $P_{\alpha}$ interacts with $\mathcal{A}'$. We will show that after $P_{\alpha}$ interacts with $\mathcal{A}$ and answers the $T^{th}$ query, the second stage of $P_{\alpha}$ also uniformly generates graphs in $\mathcal{G}_{\alpha}^{N}$.

Assuming, without loss of generality, that the answer to the $T^{th}$ query cannot be obtained from the query-answer history because this query does not provide additional information. Denote the $T^{th}$ query of $\mathcal{A}$ as $q_T=(v_T,i_T)$ and consider all actions of the process $P_{1}$:

\begin{itemize}
    \item 
    \textbf{(Case 1)} $\mathbf{i_T \in \{3,4,5,6\}}$\textbf{, and }$\mathbf{v_T}$ \textbf{ in } $h$:
    
    Assume, without loss of generality, that $i_{T} = 3$ and denote $\overline{p}=\sigma^{\text{2nd}}(p_{v_T})$. 
    The probability of $P_1$ connecting $v_{T}$ to any vertex is independent of the specific order of vertices on the cycle but depends on the labeling of the vertices. 
    After considering all possible connecting edges carried out in the second stage following the interaction with $\mathcal{A}'$, it becomes evident that the only vertices in $h$ to which $v_{T}$ can connect are those in $X_{\overline{p}, 4}$. 
    In any potential arrangement of the vertices on the cycle, there will be exactly $(N / 10)-n_{\overline{p}}$ vertices labeled $\overline{p}$ and available for connection to $v_{T}$. 
    This implies that the probability of $v_{T}$ being connected to a vertex in $X_{\overline{p}, 4}$ is $\frac{\left|X_{\overline{p}, 4}\right|}{\left|X_{\overline{p}, 4}\right|+(N / 10)-n_{\overline{p}}}$. 
    Furthermore, when $v_{T}$ is connected to a vertex in $X_{\overline{p}, 4}$, this vertex is uniformly distributed within $X_{\overline{p}, 4}$. 
    Similarly, when connected to a vertex not in $h$, this vertex is uniformly distributed among the vertices not in $h$. 
    These probabilities align with the definitions in $P_1$. 
    Therefore, in Case 1, the induction step holds for $P_1$.

    \item
    \textbf{(Case 2)} $\mathbf{i_T \in \{1,2\}}$\textbf{, and }$\mathbf{v_T}$ \textbf{in} $h$: 
    
    Assume, without loss of generality, that $i_{T} = 1$ and denote $\overline{p}$ as $\sigma^{\text{1st}}(p_{v_T})$.
    In any valid embedding of the edges in $h$ onto the cycle, it is evident that $v_{T}$ can be adjacent with another vertex $u$ in $h$ only if $u$ belongs to $X_{\overline{p}, 2}$. 
    Moreover, when $v_{T}$ is adjacent to a vertex in $h$, this vertex is uniformly distributed within $X_{\overline{p}, 2}$. 
    If $v_{T}$ is adjacency with another vertex $u$ not in $h$, it is evident that the number of vertices labeled $\overline{p}$ but not in $h$ is $(N / 10)-n_{\overline{p}}$.
    Consequently, the probability of $v_{T}$ being adjacent to some $u \in X_{\overline{p}, 2}$ is $\frac{\left|X_{\overline{p}, 2}\right|}{\left|X_{\overline{p}, 2}\right|+(N / 10)-n_{\overline{p}}}$, and the probability of it being adjacent to a vertex not in $h$ is $\frac{(N / 10)-n_{\overline{p}}}{\left|X_{\overline{p}, 2}\right|+(N / 10)-n_{\overline{p}}}$. 
    These probabilities align with the definitions in $P_1$. 
    Therefore, in this case, the induction step holds for $P_1$.
    \item
    \textbf{(Case 3)} $\mathbf{v_T}$ \textbf{is not in} $h$: 
    
    We can reduce this case to case 1 and 2, provided that the label of $v_{T}$ is selected with the appropriate probability. 
    In the second stage, each vertex is randomly assigned label based on the proportion of missing vertices with that label. 
    This essentially follows the assignment rule outlined in case (1) in the first stage of $P_1$.  
\end{itemize}

For $P_2$, we omit the proof since it is similar to the argument in $P_1$, and this lemma follows.

\end{proof} 

\subsection{Proof of Lemma~\ref{lemma:statistical_diff}}\label{sec:4_3}

We may assume that $\mathcal{A}$ does not make a query whose answer can be obtained from its query answer history $h$ since such a query does not update the $h$. 
Then, we begin the proof by proving the following proposition.  
\begin{proposition} (\cite{goldreich1997property}, Claim in lemma 7.4) 
Both in $\mathbf{D}_{1}^{\mathcal{A}}$ and in $\mathbf{D}_{2}^{\mathcal{A}}$, the total probability mass assigned to query-answer histories in which for some $t \leq T$ a vertex in $h$ is returned as an answer to the $t^{\text {th}}$ query is at most $10 \delta^{2}$. 
\end{proposition} 
\begin{proof} 
We begin the proof by claiming that the probability of the event that the answer in the $t^{\text {th}}$ query is a vertex in $h$ is at most $20(t-1)/N$ for every $t\leq T$.  
The statement can be derived by observing that there are at most $2(t-1)$ vertices in $h$, and uses the definition of both processes.  
Then, the probability that the event occurs in an arbitrary query-answer history of length $T$ is at most $\sum_{t=1}^{\delta \sqrt{N}} \frac{20(t-1)}{N}<10 \delta^{2}$. The proposition follows.  
\end{proof} 
From the proposition, we know that the 
edges in $h$ will not form a cycle with probability at least $1-10\delta^2$.
This event implies that for each query, these two processes pick a random vertex uniformly among the vertices, not in $h$. 
In addition, $\mathcal{A}$’s queries can only depend on the previous query-answer histories. 
Therefore, the distributions of the query-answer histories for these two processes
are identical, except if we found a cycle, which happens with probability at most $10\delta^2$. 
\cref{lemma:statistical_diff} follows.

\bibliographystyle{unsrt} 
\bibliography{ref}

\end{document}